\documentclass{article}

\usepackage{fullpage}
\usepackage{amsthm,amssymb,amsmath}  
\usepackage{xspace,enumerate}
\usepackage[utf8]{inputenc}
\usepackage{thmtools}
\usepackage{thm-restate}
\usepackage{authblk}

  \theoremstyle{plain}
  \newtheorem{theorem}{Theorem}
  \newtheorem{lemma}{Lemma}  
    
  \newtheorem{fact}{Fact}
  \newtheorem{observation}{Observation}
  \theoremstyle{definition}
  \newtheorem{definition}{Definition}
  
  \newtheorem{example}{Example}
  \newtheorem{remark}[definition]{Remark}
 
  \newtheorem*{claim}{Claim}
 
\title{String Periods in the Order-Preserving Model}

\author[1]{Garance Gourdel\footnote{A part of this work was done during the author's internship at University of Warsaw, Poland.}$^{,}$}
\author[2]{Tomasz Kociumaka\footnote{Supported by the Polish National Science Center, grant no.\ 2014/13/B/ST6/00770.}$^{,}$}
\author[2]{Jakub Radoszewski$^{\dagger,}$}
\author[2]{Wojciech Rytter$^{\dagger,}$}
\author[3]{Arseny Shur}
\author[2]{Tomasz Waleń$^{\dagger,}$}

\affil[1]{ENS Paris-Saclay, Cachan, France\\
\texttt{garance.gourdel@ens-paris-saclay.fr}}
\affil[2]{Institute of Informatics, University of Warsaw, Warsaw, Poland\\
    \texttt{[kociumaka,jrad,rytter,walen]@mimuw.edu.pl}}
\affil[3]{Ural Federal University, Ekaterinburg, Russia\\
    \texttt{arseny.shur@urfu.ru}
}

\date{\vspace{-5ex}}

\usepackage[utf8]{inputenc}
\usepackage[T1]{fontenc}
\usepackage{hyperref}

\usepackage{stmaryrd}
\usepackage{enumitem}
\setlist{leftmargin=*}
\usepackage[noend, noline, linesnumbered, ruled]{algorithm2e}
\SetKwComment{Comment}{$\triangleright$\ \rm}{}
\usepackage{amsmath}
\usepackage{amssymb,amsthm}
\usepackage{tikz}

\newcommand{\PP}{\mathit{PER}}
\newcommand{\TT}{P}

\begin{document}
\newcommand{\floor}[1]{\left\lfloor #1 \right\rfloor}
\newcommand{\ceil}[1]{\left\lceil #1 \right\rceil}

\newcommand{\NN}{\mathcal{N}}
\newcommand{\lb}{\llbracket}
\newcommand{\rb}{\rrbracket}
\renewcommand{\AA}{\mathcal{A}}
\newcommand{\BB}{\mathcal{B}}
\newcommand{\CC}{\mathcal{C}}

\newcommand{\LL}{\mathcal{L}}
\newcommand{\RR}{\mathcal{R}}
\newcommand{\LPP}{\mathrm{op\mbox{-}LPP}}
\newcommand{\PREF}{\mathrm{op\mbox{-}PREF}}
\newcommand{\regPREF}{\mathrm{PREF}}
\newcommand{\LCP}{\mathrm{op\mbox{-}LCP}}
\newcommand{\LCS}{\mathrm{op\mbox{-}LCS}}
\newcommand{\Squares}{\mathit{op\mbox{-}Squares}}
\newcommand{\Shifts}{\mathit{Shifts}}

\newcommand{\trace}{\mathit{trace}}
\newcommand{\sort}{\mathrm{sort}}
\newcommand{\shape}{\mathit{shape}}
\newcommand{\SHAPE}{\mathit{SH}}

\newcommand{\Div}{\mathit{Div}}

\newcommand{\per}{\mathsf{per}}
\newcommand{\shper}{\mathrm{sh\mbox{-}per}}

\renewcommand{\top}{\mbox{\it{maxdef}}}
\newcommand{\RoutineB}{\mathit{MaxShape}}

 \setenumerate{noitemsep,topsep=0pt,parsep=0pt,partopsep=0pt}

\maketitle
 \begin{abstract}
The order-preserving model (op-model, in short) was introduced
quite recently but has already attracted significant attention
because of its applications in data analysis.
We introduce several types of periods
in this setting (op-periods).
Then we give algorithms to compute these periods in time $O(n),\, O(n\log\log n),\,
O(n \log^2 \log n/\log \log \log n),\, O(n\log n)$ depending on the type
of periodicity.
In the most general variant the number of different periods
can be as big as $\Omega(n^2)$,
and a compact representation is needed.
Our algorithms require novel combinatorial insight into the properties of such periods.
 \end{abstract} 

  \section{Introduction}
  Study of strings in the \emph{order-preserving} model (\emph{op-model}, in short)
  is a part of the so-called non-standard stringology.
  It is focused on pattern matching and repetition discovery problems in the shapes of number sequences.
  Here the shape of a sequence is given by the relative order of its elements.
  The applications of the op-model include finding trends in time series
  which appear naturally when considering e.g.\ the stock market or
  melody matching of two musical scores; see \cite{Costas_TCS}.
  In such problems periodicity plays a crucial role.

  One of motivations is given by the following scenario.
  Consider a sequence $D$ of numbers that models a time series which is known to repeat the same shape every fixed
  period of time.
  For example, this could be certain stock market data or statistics data from a social network that is strongly dependent on the day of the week,
  i.e., repeats the same shape every consecutive week.
  Our goal is, given a fragment $S$ of the sequence $D$, to discover such repeating shapes, called here \emph{op-periods}, in $S$.
  We also consider some special cases of this setting.
  If the beginning of the sequence $S$ is synchronized with the beginning of the repeating shape in $D$, we refer to the repeating shape as to an \emph{initial} op-period.
  If the synchronization takes place also at the end of the sequence, we call the shape a \emph{full} op-period.
  Finally, we also consider \emph{sliding} op-periods that describe the case when every factor of the sequence $D$ repeats the same shape every fixed period of time.
  
  \paragraph{\bf Order-preserving model.}
  Let $\lb a..b \rb$ denote the set $\{a,\ldots,b\}$.
  We say that two strings $X=X[1] \ldots X[n]$ and $Y=Y[1] \ldots Y[n]$ over an integer alphabet
  are \emph{order-equivalent} (\emph{equivalent} in short), written $X\approx Y$, iff
  $\forall_{ i,j \in \lb 1 .. n\rb}\ \ X[i] <  X[j]\Leftrightarrow Y[i] < Y[j]$.
  \begin{example}
    $5\, 2\, 7\, 5\, 1\, 3\, 10\, 3\, 5 \approx 6\, 4\, 7\, 6\, 3\, 5\, 9\, 5\, 6$.
  \end{example}
  Order-equivalence is a special case of a substring consistent equivalence relation (SCER) that was defined in \cite{DBLP:journals/tcs/MatsuokaAIBT16}.

  For a string $S$ of length $n$, we can create a new string $X$ of length $n$ such that $X[i]$
  is equal to the number of distinct symbols in $S$ that are not greater than $S[i]$.
  The string $X$ is called the \emph{shape} of $S$ and is denoted by $\shape(S)$.
  It is easy to observe that two strings $S,T$ are order-equivalent if and only if they have the same shape. 
  \begin{example}
    $\shape(5\, 2\, 7\, 5\, 1\, 3\, 10\, 3\, 5)= \shape(6\, 4\, 7\, 6\, 3\, 5\, 9\, 5\, 6)=4\,2\,5\,4\,1\,3\,6\,3\,4$.
  \end{example}

  \paragraph{\bf Periods in the op-model.}
  We consider several notions of periodicity in the op-model, illustrated by Fig.~\ref{f:types}. 
  We say that a string $S$ has a (general) \emph{op-period} $p$ with \emph{shift} $s \in \lb 0 .. p-1 \rb$
  if and only if $p<|S|$ and $S$ is a factor of a string $V_1V_2 \cdots V_k$ such that:
  $$|V_1|=\dots=|V_k|=p,\quad V_1 \approx \dots \approx V_k,\quad\text{and } S[s+1..|S|]\text{ is a prefix of }V_2\cdots V_k.$$

  \noindent
  The \emph{shape} of the op-period is $\shape(V_1)$.
  One op-period $p$ can have several shifts; to avoid ambiguity, we sometimes denote the op-period as $(p,s)$.
  We define $\Shifts_p$ as the set of all shifts of the op-period $p$.
  
  An op-period $p$ is called \emph{initial} if $0\in \Shifts_p$, \emph{full} if it is initial and $p$ divides $|S|$, and \emph{sliding} if $\Shifts_p = \lb 0..p-1 \rb$.
  Initial and sliding op-periods are particular cases of block-based and sliding-window-based periods for SCER,
  both of which were introduced in~\cite{DBLP:journals/tcs/MatsuokaAIBT16}.

\begin{figure}[!htb]
\begin{center}
\begin{tikzpicture}

\foreach \x [count=\i] in {0, 0, 3, 2, 1, 1, 3, 2, 1, 1, 4, 3} {
  \node at (\i * 0.5, 0) {\small \x};
}

\begin{scope}
\clip (0.24, -0.8) rectangle (6.26, 0.8);
\foreach \p/\s [count=\i] in {4/0} {
  \foreach \c in {0, 1, 2} {
    \draw (\s * 0.5 + 0.25 + \c * \p * 0.5,\i * 0.25 + 0.2) -- +(\p*0.5, 0);
    \draw (\s * 0.5 + 0.25 + \c * \p * 0.5,\i * 0.25 + 0.2) -- +(0, -0.15);
    \draw (\s * 0.5 + 0.25 + \c * \p * 0.5 + \p * 0.5,\i * 0.25 + 0.2) -- +(0, -0.15);
  }
}

\foreach \p/\s [count=\i] in {4/-2, 4/-1} {
  \foreach \c in {0, 1, 2, 3} {
    \draw (\s * 0.5 + 0.25 + \c * \p * 0.5, -\i * 0.25 - 0.2) -- +(\p*0.5, 0);
    \draw (\s * 0.5 + 0.25 + \c * \p * 0.5, -\i * 0.25  - 0.2) -- +(0, 0.15);
    \draw (\s * 0.5 + 0.25 + \c * \p * 0.5 + \p * 0.5,-\i * 0.25 - 0.2) -- +(0, 0.15);
  }
}
\end{scope}

\begin{scope}[xshift=7cm]
\foreach \x [count=\i] in {1, 1, 2, 5, 1, 1, 3, 4, 1, 1, 2, 4} {
  \node at (\i * 0.5, 0) {\small \x};
}

\begin{scope}
\clip (0.24, -0.8) rectangle (6.26, 0.85);
\foreach \p/\s [count=\i] in {4/0, 4/1} {
  \foreach \c in {-1, 0, 1, 2, 3} {
    \draw (\s * 0.5 + 0.25 + \c * \p * 0.5,\i * 0.25 + 0.2) -- +(\p*0.5, 0);
    \draw (\s * 0.5 + 0.25 + \c * \p * 0.5,\i * 0.25 + 0.2) -- +(0, -0.15);
    \draw (\s * 0.5 + 0.25 + \c * \p * 0.5 + \p * 0.5,\i * 0.25 + 0.2) -- +(0, -0.15);
  }
}

\foreach \p/\s [count=\i] in {4/-2, 4/-1} {
  \foreach \c in {0, 1, 2, 3} {
    \draw (\s * 0.5 + 0.25 + \c * \p * 0.5, -\i * 0.25 - 0.2) -- +(\p*0.5, 0);
    \draw (\s * 0.5 + 0.25 + \c * \p * 0.5, -\i * 0.25  - 0.2) -- +(0, 0.15);
    \draw (\s * 0.5 + 0.25 + \c * \p * 0.5 + \p * 0.5,-\i * 0.25 - 0.2) -- +(0, 0.15);
  }
}
\end{scope}
\end{scope}

\end{tikzpicture}
\end{center}
\caption{
  The string to the left has op-period 4 with three shifts: $\Shifts_4=\lb 0..0 \rb \cup \lb 2..3 \rb$.
  Due to the shift 0, the string has an initial---therefore, a full---op-period 4.
  The string to the right has op-period 4 with all four shifts: $\Shifts_4=\lb 0..3 \rb$.
  In particular, 4 is a sliding op-period of the string.
Notice that both strings (of length $n=12$) have
(general, sliding) periods 4, but none of them has the order-border (in the sense of
\cite{DBLP:journals/ipl/KubicaKRRW13}) of length $n-4$.
}\label{f:types}
\end{figure}
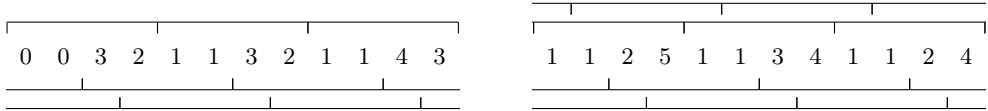

  \paragraph{\bf Models of periodicity.}
  In the standard model, a string $S$ of length $n$ has a period $p$ iff $S[i]=S[i+p]$ for all $i=1,\ldots,n-p$. The famous periodicity lemma of Fine and Wilf~\cite{fw:65} states that a ``long enough'' string with periods $p$ and $q$ has also the period $\gcd(p,q)$. The exact bound of being ``long enough'' is $p+q-\gcd(p,q)$. This result was generalized to arbitrary number of periods \cite{DBLP:journals/tcs/CastelliMR99,DBLP:journals/ita/Justin00,TijdemanZamboni03}.

Periods were also considered in a number of non-standard models. Partial words, which are strings with don't care symbols, possess quite interesting Fine--Wilf type properties, including probabilistic ones; see \cite{DBLP:journals/tcs/BerstelB99,DBLP:journals/iandc/Blanchet-SadriBS08,DBLP:journals/tcs/Blanchet-SadriH02,ShurGamzova04,DBLP:conf/mfcs/ShurK01,DBLP:journals/fuin/IdiatulinaS14}. In Section~\ref{s:FW}, we make use of periodicity graphs introduced in \cite{ShurGamzova04,DBLP:conf/mfcs/ShurK01}. In the abelian (jumbled) model, a version of the periodicity lemma was shown in \cite{DBLP:journals/eatcs/ConstantinescuI06} and extended in \cite{DBLP:journals/ita/Blanchet-SadriSTV13}. Also, algorithms for computing three types of periods analogous to full, initial, and general op-periods were designed \cite{DBLP:journals/ipl/CrochemoreIKKPRRTW13,DBLP:journals/dam/FiciLLP14,DBLP:journals/dam/FiciLLPS16,DBLP:conf/stacs/KociumakaRR13,DBLP:journals/jcss/KociumakaRR17,DBLP:conf/macis/KociumakaRW15}. In the computation of full and initial op-periods we use some number-theoretic tools initially developed in \cite{DBLP:conf/stacs/KociumakaRR13,DBLP:journals/jcss/KociumakaRR17}. Remarkably, the fastest known algorithm for computing general periods in the abelian model has essentially quadratic time complexity \cite{DBLP:journals/ipl/CrochemoreIKKPRRTW13,DBLP:conf/macis/KociumakaRW15}, whereas for the general op-periods we design a much more efficient solution. A version of the periodicity lemma for the parameterized model was proposed in \cite{DBLP:journals/dam/ApostolicoG08}.

Op-periods were first considered in \cite{DBLP:journals/tcs/MatsuokaAIBT16} where initial and sliding op-periods were introduced and direct generalizations of the Fine--Wilf property to these kinds of op-periods were developed. A few distinctions between the op-periods and periods in other models should be mentioned. First, ``to have a period 1'' becomes a trivial property in the op-model. Second, all standard periods of a string have the ``sliding'' property; the first string in Fig.~\ref{f:types} demonstrates that this is not true for op-periods. The last distinction concerns borders. A standard period $p$ in a string $S$ of length $n$ corresponds to a \emph{border} of $S$ of length $n-p$, which is both a prefix and a suffix of $S$. In the order-preserving setting, an analogue of a border is an \emph{op-border}, that is, a prefix that is equivalent to the suffix of the same length. Op-borders have properties similar to standard borders and can be computed in $O(n)$ time \cite{DBLP:journals/ipl/KubicaKRRW13}. However, it is no longer the case that a (general, initial, full, or sliding) op-period must correspond to an op-border; see \cite{DBLP:journals/tcs/MatsuokaAIBT16}.

  \paragraph{\bf Previous algorithmic study of the op-model.}
  The notion of order-equivalence was introduced in \cite{Costas_TCS,DBLP:journals/ipl/KubicaKRRW13}.
	(However, note the related combinatorial studies, originated in \cite{ElizaldeNoy03}, on containment/avoidance of shapes in permutations.)
  Both \cite{Costas_TCS,DBLP:journals/ipl/KubicaKRRW13} studied pattern matching in the op-model (op-pattern matching)
  that consists in identifying all consecutive factors of a text that are order-equivalent
  to a given pattern.
  We assume that the alphabet is integer and, as usual, that it is polynomially bounded with respect to the length of the string, which
  means that a string can be sorted in linear time (cf.\ \cite{DBLP:books/daglib/0023376}).
  Under this assumption, for a text of length $n$ and a pattern of length $m$, \cite{Costas_TCS} solve the op-pattern matching problem in
  $O(n+m\log{m})$ time and \cite{DBLP:journals/ipl/KubicaKRRW13} solve it in $O(n+m)$ time.
  Other op-pattern matching algorithms were presented in \cite{DBLP:conf/isaac/BelazzouguiPRV13,DBLP:journals/ipl/ChoNPS15}.
  
  An index for op-pattern matching based on the suffix tree was developed in \cite{DBLP:journals/tcs/CrochemoreIKKLP16}.
  For a text of length $n$ it uses $O(n)$ space and answers op-pattern matching queries for a pattern of length $m$
  in optimal, $O(m)$ time (or $O(m+\mathit{Occ})$ time if we are to report all $\mathit{Occ}$ occurrences).
  The index can be constructed in $O(n \log \log n)$ expected time or $O(n \log^2 \log n/\log \log \log n)$ worst-case time.
  We use the index itself and some of its applications from \cite{DBLP:journals/tcs/CrochemoreIKKLP16}.

  Other developments in this area include a multiple-pattern matching algorithm for the op-model \cite{Costas_TCS},
  an approximate version of op-pattern matching \cite{DBLP:journals/tcs/GawrychowskiU16},
  compressed index constructions \cite{DBLP:conf/stringology/ChhabraKT15,DBLP:conf/dcc/DecaroliGM17},
  a small-space index for op-pattern matching that supports only short queries \cite{gagie_et_al:LIPIcs:2017:7872},
  and a number of practical approaches \cite{DBLP:conf/stringology/CantoneFK15,DBLP:journals/spe/ChhabraFKT17,DBLP:conf/spire/ChhabraGT15,DBLP:journals/ipl/ChhabraT16,DBLP:conf/aaim/FaroK16}.

\paragraph{\bf Our results.}
We give algorithms to compute:
\begin{itemize}
\item all full op-periods in $O(n)$ time;
\item the smallest non-trivial initial op-period in $O(n)$ time;
\item all initial op-periods in $O(n \log \log n)$ time;
\item all sliding op-periods in $O(n \log \log n)$ expected time or $O(n \log^2 \log n / \log \log \log n)$ worst-case time (and linear space);
\item all general op-periods with all their shifts (compactly represented) in $O(n\log n)$ time and space.
The output is the family of sets $\Shifts_p$ represented as unions of
disjoint intervals. 
The total number of intervals, over all $p$, is $O(n\log n)$.
\end{itemize}
In the combinatorial part, we characterize the Fine--Wilf periodicity property (aka interaction property) in the op-model in the case of coprime periods. This result is at the core of the linear-time algorithm for the smallest initial op-period.

  \paragraph{Structure of the paper.}
  Combinatorial foundations of our study are given in Section~\ref{s:FW}. Then in Section~\ref{sec:toolbox} we recall known algorithms and data structures for the op-model and develop further algorithmic tools. The remaining sections are devoted to computation of the respective types of op-periods: full and initial op-periods in Section~\ref{sec:full_initial}, the smallest non-trivial initial op-period in Section~\ref{sec:sm_initial}, all (general) op-periods in Section~\ref{sec:general}, and sliding op-periods in Section~\ref{sec:sliding}.

\section{Fine--Wilf Property for Op-Periods} \label{s:FW}

The following result was shown as Theorem 2 in \cite{DBLP:journals/tcs/MatsuokaAIBT16}. Note that if $p$ and $q$ are coprime, then the conclusion is void, as every string has the op-period 1.

\begin{theorem}[\cite{DBLP:journals/tcs/MatsuokaAIBT16}] \label{t:fwd2}
Let $p>q>1$ and $d=\gcd(p,q)$. If a string $S$ of length $n\ge p+q-d$ has initial op-periods $p$ and $q$, it has initial op-period $d$.
Moreover, if $S$ has length $n \ge p+q-1$ and sliding op-periods $p$ and $q$, it has sliding op-period $d$.
\end{theorem}

\noindent
The aim of this section is to show a periodicity lemma in the case that $\gcd(p,q)=1$.

\subsection{Preliminary Notation}
For a string $S$ of length $n$, by $S[i]$ (for $1 \le i \le n$) we denote the $i$th letter of $S$ and by $S[i..j]$ we denote a \emph{factor} of $S$ equal to $S[i] \ldots S[j]$. If $i>j$, $S[i..j]$ denotes the empty string $\varepsilon$.

A string which is strictly increasing, strictly decreasing, or constant, is called \emph{strictly monotone}. 
A \emph{strictly monotone op-period} of $S$ is an op-period with a strictly monotone shape.
Such an op-period is called increasing (decreasing, constant) if so is its shape. 
Clearly, any divisor of a strictly monotone op-period is a strictly monotone op-period as well.
A string $S$ is \emph{2-monotone} if $S=S_1S_2$, where $S_1,S_2$ are strictly monotone in the same direction.

Below we assume that $n>p>q>1$. Let a string $S=S[1..n]$ have op-periods $(p,i)$ and $(q,j)$.
If there exists a number $k\in \lb 1..n-1 \rb$ such that $k\bmod p = i$ and $k\bmod q = j$, we say that these op-periods are \emph{synchronized} and $k$ is a \emph{synchronization point} (see Fig.~\ref{f:synch}).

\begin{figure}[!htb]
\centerline{
\includegraphics[trim={33 715 270 60}, clip]{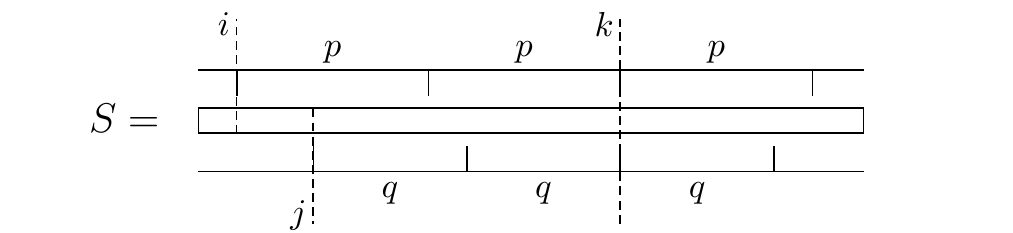} }
\caption{Op-periods $(p,i)$ and $(q,j)$ synchronized at position $k$.}\label{f:synch}
\end{figure}

\begin{remark}
The proof of Theorem~\ref{t:fwd2} can be easily adapted to prove the following.
\end{remark}

\begin{theorem} \label{t:fwd3}
Let $p>q>1$ and $d=\gcd(p,q)$. If op-periods $p$ and $q$ of a string $S$ of length $n\ge p+q-1$ are synchronized, then $S$ has op-period $d$, synchronized with them.
\end{theorem}

\subsection{Periodicity Theorem For Coprime Periods}
For a string $S$, by $\trace(S)$ we denote a string $X$ of length $|S|-1$ over the alphabet
$\{\texttt{+},\texttt{0},\texttt{-}\}$ such that:
  $$
    X[i]=\left\{
\begin{array}{ll}
\texttt{+} & \mbox{if }S[i]<S[i+1]\\
\texttt{0} & \mbox{if }S[i]=S[i+1]\\
\texttt{-} & \mbox{if }S[i]>S[i+1].
\end{array}
\right.
  $$

\begin{observation} \label{obs:trace}
\begin{enumerate}[label={\rm(\arabic*)}]
\item\label{oo1}
A string is strictly monotone iff its trace is a unary string.
\item\label{oo2}
If $S$ has an op-period $p$ with shift $i$, then $\trace(S)$ ``almost'' has a period $p$, namely, $\trace(S)[j]=\trace(S)[k]$ for any $j,k\in\lb 1..n-1 \rb$ such that $j=k\pmod p$ and $j\ne i\pmod p$. (This is because both $\trace(S)[j]$ and $\trace(S)[k]$ equal the sign of the difference between the same positions of the shape of the op-period of $S$.)
\end{enumerate}
\end{observation}

\begin{example}
  Consider the string 7\,5\,8\,1\,4\,6\,2\,4\,5. It has an op-period $(3,1)$ with shape 2\,3\,1.
  The trace of this string is:

\centerline{\texttt{\textcolor{white!70!black}{-}\ +\ -\ 
\textcolor{white!70!black}{+}\ +\ -\ \textcolor{white!70!black}{+}\ +}
}

\noindent
The positions giving the remainder 1 modulo 3 are shown in gray; the sequence of the remaining positions is periodic.
\end{example}
To study traces of strings with two op-periods, we use \emph{periodicity graphs} (see Fig.~\ref{f:graph} below) very similar to those introduced in \cite{ShurGamzova04,DBLP:conf/mfcs/ShurK01} for the study of partial words with two periods. The periodicity graph $G(n,p,i,q,j)$ represents all strings $S$ of length $n{+}1$ having the op-periods $(p,i)$ and $(q,j)$. Its vertex set $\lb 1..n \rb$ is the set of positions of the trace $\trace(S)$. Two positions are connected by an edge iff they contain equal symbols according to Observation~\ref{obs:trace}\ref{oo2}. For convenience, we distinguish between $p$- and $q$-edges, connecting positions in the same residue class modulo $p$ (resp., modulo $q$). The construction of $G(n,p,i,q,j)$ is split in two steps: first we build a \emph{draft graph} $H(n,p,q)$ (see Fig.~\ref{f:graph},a), containing all $p$- and $q$-edges for each residue class, and then delete all edges of the orange clique corresponding to the $i$th class modulo $p$ and all edges of the blue clique corresponding to the $j$th class modulo $q$ (see Fig.~\ref{f:graph},b,c). If some vertices $k,l$ belong to the same connected component of $G=G(n,p,i,q,j)$, then $\trace(S)[k]=\trace(S)[l]$ for every string $S$ corresponding to $G$. In particular, if $G$ is connected, then $\trace(S)$ is unary and $S$ is strictly monotone by Observation~\ref{obs:trace}\ref{oo1}. 

\begin{figure}[htpb]
\begin{center}
\includegraphics[scale=0.83, trim={47 697 390 28}, clip]{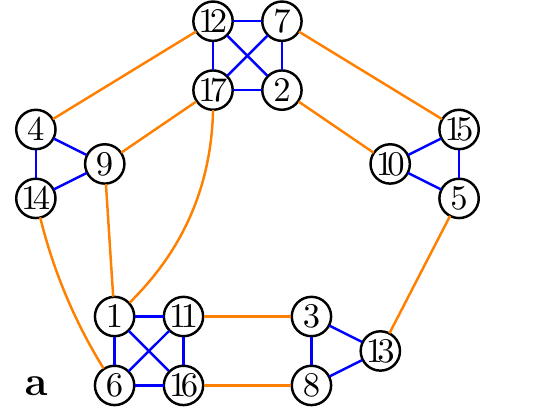}
\includegraphics[scale=0.83, trim={47 697 390 28}, clip]{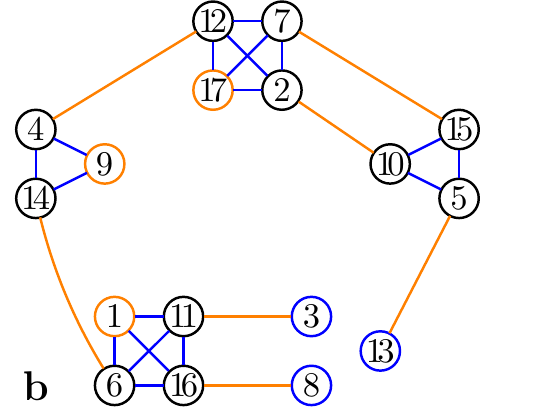}
\includegraphics[scale=0.83, trim={47 697 410 28}, clip]{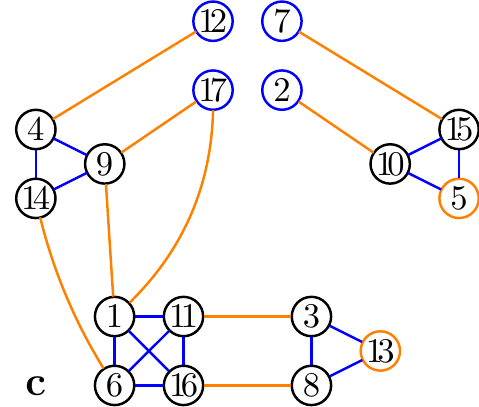}
\end{center}
\caption{Examples of periodicity graph: (a) draft graph $H(17,8,5)$; (b) periodicity graph $G(17,8,1,5,3)$; (c) periodicity graph $G(17,8,5,5,2)$. Orange/blue are $p$-edges (resp., $q$-edges) and the vertices equal to $i$ modulo $p$ (resp., to $j$ modulo $q$).}\label{f:graph}
\end{figure}

\begin{example}
The graph $G(17,8,1,5,3)$ in Fig.~\ref{f:graph},b is connected, so all strings having this graph are strictly monotone. On the other hand, some strings with the graph $G(17,8,5,5,2)$ in Fig.~\ref{f:graph},c have no monotonicity properties. Thus, the string $6\,1\!8\,2\,1\!5\,1\!7\,3\,1\!6\,1\,5\,1\!4\,4\,7\,8\,1\!0\,1\!3\,9\,1\!1\,1\!2$ of length 18 indeed has the op-period 8 with shift 5 (and shape $2\,8\,1\,4\,7\,3\,5\,6$) and the op-period 5 with shift 2 (and shape $1\,3\,5\,2\,4$).
\end{example}

It turns out that the existence of two coprime op-periods makes a string ``almost'' strictly monotone.

\begin{theorem}\label{t:fw1}
Let $S$ be a string of length $n$ that has coprime op-periods $p$ and $q$ with shifts $i$ and $j$, respectively, such that $n > p > q > 1$. Then:
\begin{enumerate}[label={\rm(\alph*)}]
\item\label{aa} if $n>pq$, then $S$ has a strictly monotone op-period $pq$;
\item\label{bb} if $2p<n\le pq$ and the op-periods are synchronized, then $S$ is 2-monotone;
\item\label{cc} if $p{+}q<n\le 2p$ and the op-periods are synchronized, then $(q,j)$ is a strictly monotone op-period of $S$;
\item\label{dd} if $n>\max\{2p,p{+}2q\}$ and the op-periods are not synchronized, then $S$ is strictly monotone;
\item\label{ee} if $n>2p$, the op-periods are not synchronized, and $p$ is initial, then $S$ is strictly monotone;
\item\label{ff} if $p{+}q<n\le 2p$ and $p$ is initial, then $(q,j)$ is a strictly monotone op-period of $S$.
\end{enumerate}
\end{theorem}
\begin{proof}
Take a string $S$ of length $n$ having op-periods $p$ (with shift $i$) and $q$ (with shift $j$). Let $n'=n-1$. Consider the draft graph $H(n',p,q)$ (see Fig.~\ref{f:graph},a). It consists of $q$ $q$-cliques (numbered from 0 to $q-1$ by residue classes modulo $q$) connected by some $p$-edges. If $n'=p+q$, there are exactly $q$ $p$-edges, which connect $q$-cliques in a cycle due to coprimality of $p$ and $q$. Thus we have a cyclic order on $q$-cliques: for the clique $k$, the next one is $(k{+}p)\bmod q$. The number of $p$-edges connecting neighboring cliques increases with the number of vertices: if $n'\ge 2p$, every vertex has an adjacent $p$-edge, and if $n'\ge p+2q$, every $q$-clique is connected to the next $q$-clique by at least two $p$-edges.

To obtain the periodicity graph $G(n',p,i,q,j)$, one should delete all edges of the $i$th $p$-clique and the $j$th $q$-clique from $H(n',p,q)$. First consider the effect of deleting $p$-edges. If the $i$th $p$-clique has at least three vertices, then after the deletion each $q$-clique will still be connected to the next one. Indeed, if we delete edges between $i$, $i{+}p$, and $i{+}2p$, then there are still the edges $(i{+}q,i{+}p{+}q)$ and $(i{+}p{-}q,i{+}2p{-}q)$, connecting the corresponding $q$-cliques. If the $p$-clique has a single edge, its deletion will break the connection between two neighboring $q$-cliques if they were connected by a single edge. This is not the case if $n'\ge p{+}2q$, but may happen for any smaller $n'$; see Fig.~\ref{f:graph},c, where $n'=p{+}2q{-}1$.

Now look at the effect of only deleting $q$-edges from $H(n',p,q)$. If all vertices in the $j$th $q$-clique have $p$-edges (this holds for any $j$ if $n'\ge 2p$), the graph after deletion remains connected; if not, it consists of a big connected component and one or more isolated vertices from the $j$th $q$-clique.

Finally we consider the cumulative effect of deleting $p$- and $q$-edges. Any synchronization point becomes an isolated vertex. In total, there are two ways of making the draft graph disconnected: break the connection between neighboring $q$-cliques distinct from the removed $q$-clique (Fig.~\ref{f:disc},a) or get isolated vertices in the removed $q$-clique (Fig.~\ref{f:disc},b). The first way does not work if $n'\ge p+2q$ (see above) or if the op-periods are synchronized (the removed $p$-edge was adjacent to the removed $q$-clique). For the second way, only synchronization points are isolated if $n'\ge 2p$ (each vertex has a $p$-edge, see above). Note that in this case all non-isolated vertices of periodicity graph are connected. Hence all positions of the trace $\trace(S)$, except for the isolated ones, contain the same symbol. So all factors of $S$ involving no isolated positions are strictly monotone (in the same direction). 

\begin{figure}[!htb]
\centerline{
\includegraphics[trim={47 697 390 28}, clip]{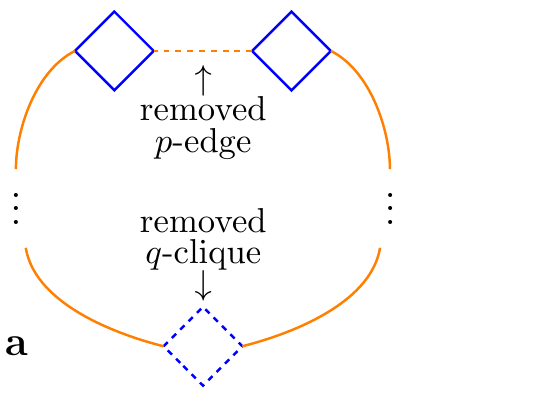}
\includegraphics[trim={47 697 390 28}, clip]{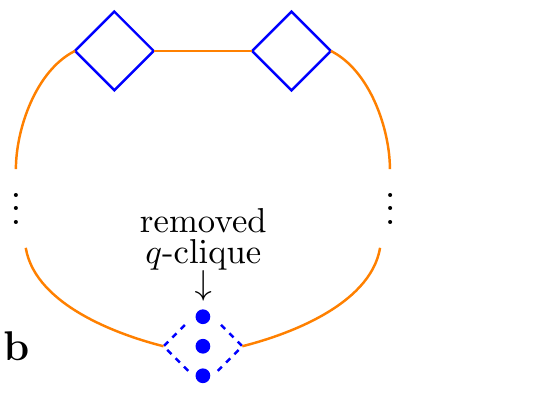} }
\caption{Disconnecting a draft graph: (a) removing the only edge between neighboring $q$-cliques distinct from the removed $q$-clique; (b) getting isolated vertices in the removed $q$-clique.}\label{f:disc}
\end{figure}

At this point all statements of the theorem are straightforward:
\begin{itemize}
\item[(a,b)] all synchronization points are equal modulo $pq$ by the Chinese Remainder Theorem;
\item[(c)] all isolated positions are equal modulo $q$;
\item[(d)] the condition on $n$ excludes both ways to disconnect the draft graph;
\item[(e,f)] for the initial op-period, $i=p$; if $n\le 2p$, there is no deletion of $p$-edges; if $n>2p$, then the $q$-cliques connected by the edge $(p,2p)$ are also connected by $(p{-}q,2p{-}q)$; so only the disconnection by isolated positions is possible. \qedhere
\end{itemize}
\end{proof}

\section{Algorithmic Toolbox for Op-Model}\label{sec:toolbox}
  Let us start by recalling the \emph{encoding for op-pattern matching} (\emph{op-encoding}) from \cite{DBLP:journals/tcs/CrochemoreIKKLP16,DBLP:journals/ipl/KubicaKRRW13}.
  For a string $S$ of length $n$ and $i \in \lb 1 .. n \rb$ we define:
  $$\alpha_i(S)\text{ as the largest }j<i\text{ such that }S[j]=\max\{S[k]\ :\ k < i,\ S[k] \le S[i]\}.$$
  If there is no such $j$, then $\alpha_i(S)=0$. Similarly, we define:
  $$\beta_i(S)\text{ as the largest }j<i\text{ such that }S[j]=\min\{S[k]\ :\ k < i,\ S[k]\ge S[i]\},$$
  and $\beta_i(S)=0$ if no such $j$ exists.
  Then $(\alpha_1(S),\beta_1(S)),\ldots,(\alpha_n(S),\beta_n(S))$ is the op-encoding of $S$.
  It can be computed efficiently as mentioned in the following lemma.
  \begin{lemma}[\cite{DBLP:journals/ipl/KubicaKRRW13}]\label{lem:op_enc}
    The op-encoding of a string of length $n$ over an integer alphabet can be computed in $O(n)$ time.
  \end{lemma}

  \noindent
  The op-encoding can be used to efficiently extend a match.
 
  \begin{lemma}\label{lem:extend}
    Let $X$ and $Y$ be two strings of length $n$ and assume that the op-encoding of $X$ is known.
    If $X[1..n-1] \approx Y[1..n-1]$, one can check if $X \approx Y$ in $O(1)$ time.
  \end{lemma}
  \begin{proof}
    Let $i =\alpha_n(X)$ and $j = \beta_n(X)$.
    Lemma 3 from \cite{DBLP:journals/tcs/CrochemoreIKKLP16} asserts that,
    if $i\ne j$, then
      $$ X \approx Y \Longleftrightarrow Y[i] < Y[n] < Y[j],$$
    and otherwise,
      $$ X \approx Y \Longleftrightarrow Y[i] = Y[n] = Y[j].$$
    (Conditions involving $Y[i]$ or $Y[j]$ when $i=0$ or $j=0$ should be omitted.)
  \end{proof}

\subsection{$\PREF$ table}
For a string $S$ of length $n$, we introduce a table $\PREF[1..n]$ such that $\PREF[i]$ is the length of the longest prefix of $S[i..n]$ that is equivalent to a prefix of $S$. It is a direct analogue of the PREF array used in standard string matching (see \cite{Jewels}) and can be computed similarly in $O(n)$ time using one of the standard encodings for the op-model that were used in \cite{DBLP:journals/ipl/ChoNPS15,DBLP:journals/tcs/CrochemoreIKKLP16,DBLP:journals/ipl/KubicaKRRW13}; see lemma below.

\begin{lemma}\label{lem:PREF}
  For a string of length $n$, the $\PREF$ table can be computed in $O(n)$ time.
\end{lemma}
\begin{proof}
  Let $S$ be a string of length $n$.
  The standard linear-time algorithm for computing the $\regPREF$ table for $S$ (see, e.g., \cite{Jewels}) uses the following two properties of the table:
  \begin{enumerate}
    \item If $\regPREF[i]=k$, $i < j < i+k$, and $\regPREF[j+1-i] < i+k-j$, then $\regPREF[j]=\regPREF[j+1-i]$.
    \item If we know that $\regPREF[i] \ge k$, then $\regPREF[i]$ can be computed in $O(\regPREF[i]-k)$ time
    by extending the common prefix character by character.
  \end{enumerate}
  In the case of the $\PREF$ table, the first of these properties extends without alterations due to the transitivity of the $\approx$ relation.
  As for the second property, the matching prefix $S[1..k] \approx S[i..i+k-1]$ can be extended character by character using Lemma~\ref{lem:extend}
  provided that the op-encoding for $S$ is known.
  The op-encoding can be computed in advance using Lemma~\ref{lem:op_enc}.
\end{proof}

\noindent
Let us mention an application of the $\PREF$ table that is used further in the algorithms. We denote by $\LPP_p(S)$ (``longest op-periodic prefix'') the length of the longest prefix of a string $S$ having $p$ as an initial op-period.

\begin{lemma}\label{lem:LPP}
  For a string $S$ of length $n$, $\LPP_p(S)$ for a given $p$ can be computed in $O(\LPP_p(S)/p+1)$ time after $O(n)$-time preprocessing.
\end{lemma}
\begin{proof}
  We start by computing the $\PREF$ table for $S$ in $O(n)$ time.
  We assume that $\PREF[n+1]=0$.
  To compute $\LPP_p(S)$, we iterate over positions $i=p+1,2p+1,\ldots$ and for each of them check if $\PREF[i] \ge p$.
  If $i_0$ is the first position for which this condition is not satisfied (possibly because $i_0>n-p+1$), we have $\LPP_p(S)=i_0+\PREF[i_0]-1$.
  Clearly, this procedure works in the desired time complexity.
\end{proof}

\begin{remark}
  Note that it can be the case that $\LPP_p(S) \ne \PREF[p+1]$.
  See, e.g., the strings in Fig.~\ref{f:types} and $p=4$.
\end{remark}

\subsection{Longest Common Extension Queries}
For a string $S$, we define a \emph{longest common extension} query $\LCP(i,j)$ in the order-preserving model
as the maximum $k \ge 0$ such that $S[i..i+k-1] \approx S[j..j+k-1]$.
Symmetrically, $\LCS(i,j)$ is the maximum $k \ge 0$ such that $S[i-k+1..i] \approx S[j-k+1..j]$.

Similarly as in the standard model \cite{DBLP:books/daglib/0020103}, LCP-queries in the op-model can be answered using lowest common ancestor (LCA) queries in the op-suffix tree;
see the following lemma.

\begin{lemma}\label{lem:LCP}
  For a string of length $n$, after preprocessing in $O(n \log \log n)$ expected time or in \linebreak$O(n \log^2 \log n / \log \log \log n)$ worst-case time
  one can answer $\LCP$-queries in $O(1)$~time.
\end{lemma}
\begin{proof}
  The \emph{order-preserving suffix tree} (\emph{op-suffix tree}) that is constructed in \cite{DBLP:journals/tcs/CrochemoreIKKLP16}
  is a compacted trie of op-encodings of all the suffixes of the text.
  In $O(n \log \log n)$ expected time or $O(n \log^2 \log n / \log \log \log n)$ worst-case time
  one can construct a so-called incomplete version of the op-suffix tree
  in which each explicit node may have at most one edge whose first character label is not known.
  Fortunately, for $\LCP$-queries the labels of the edges are not needed; the only required
  information is the depth of each explicit node and the location of each suffix.
  Therefore, for this purpose the incomplete op-suffix tree can be treated as a regular suffix tree
  and preprocessed using standard lowest common ancestor data structure that requires additional
  $O(n)$ preprocessing and can answer queries in $O(1)$ time \cite{LCA}.
\end{proof}

\subsection{Order-preserving Squares}
The factor $S[i..i+2p-1]$ is called an order-preserving square (\emph{op-square}) iff $S[i..i+p-1] \approx S[i+p..i+2p-1]$.
For a string $S$ of length $n$, we define the set

$$\Squares_p = \{i \in \lb 1..n-2p+1 \rb\,:\,S[i..i+2p-1]\text{ is an op-square}\}.$$

Op-squares were first defined in \cite{DBLP:journals/tcs/CrochemoreIKKLP16} where an algorithm computing all the sets $\Squares_p$
for a string of length $n$ in $O(n \log n + \sum_p |\Squares_p|)$ time was shown.

We say that an op-square $S[i..i+2p-1]$ is \emph{right shiftable} if
$S[i+1..i+2p]$ is an op-square and \emph{right non-shiftable} otherwise.
Similarly, we say that the op-square is \emph{left shiftable} if $S[i-1..i+2p-2]$ is an op-square
and \emph{left non-shiftable} otherwise.
Using the approach of~\cite{DBLP:journals/tcs/CrochemoreIKKLP16}, one can show the following lemma.

\begin{lemma}\label{lem:op_sq}
  All the (left and right) non-shiftable op-squares in a string of length $n$ can be computed in $O(n \log n)$ time.
\end{lemma}
\begin{proof}
  We show the algorithm for right non-shiftable op-squares; the computations for left non-shiftable op-squares are symmetric.

  Let $S$ be a string of length $n$.
  An op-square $S[i..i+2p-1]$ is called \emph{right non-extendible} if $i+2p-1=n$
  or $S[i..i+p] \not\approx S[i+p..i+2p]$.
  We use the following claim.
  \begin{claim}[See Lemma 18 in \cite{DBLP:journals/tcs/CrochemoreIKKLP16}]\label{clm:non_ext}
    All the right non-extendible op-squares in a string of length $n$ can be computed in $O(n \log n)$ time.
  \end{claim}
  Note that a right non-shiftable op-square is also right non-extendible, but the converse is not necessarily true.
  Thus it suffices to filter out the op-squares that are right shiftable.
  For this, for a right non-extendible op-square $S[i..i+2p-1]$ we need to check if $\LCP(i+1,i+p+1)<p$.
  This condition can be verified in $O(1)$ time after $o(n \log n)$-time preprocessing using Lemma~\ref{lem:LCP}.
\end{proof}

\section{Computing All Full and Initial Op-Periods}\label{sec:full_initial}
For a string $S$ of length $n$, we define $\PREF'[i]$ for $i=0,\ldots,n$ as:
$$\PREF'[i]=\left\{
\begin{array}{cl}
n & \mbox{if }\PREF[i+1]=n-i\\
\PREF[i+1] & \mbox{otherwise.}
\end{array}
\right.
$$
Here we assume that $\PREF[n+1]=0$.
In the computation of full and initial op-periods we heavily rely on this table according to the following obvious observation.
\begin{observation}\label{obs:ini}
  $p$ is an initial op-period of a string $S$ of length $n$ if and only if $\PREF'[ip] \ge p$ for all $i=1,\ldots,\lfloor n/p\rfloor$.
\end{observation}

\subsection{Computing Initial Op-Periods}\label{sec:ini_all}
Let us introduce an auxiliary array $\TT[0..n]$ such that:
$$\TT[p]=\min\{\PREF'[ip]\,:\,i=1,\ldots,\floor{n/p}\}.$$
Straight from Observation~\ref{obs:ini} we have:

\begin{observation}
  $p$ is an initial period of $S$ if and only if $\TT[p] \ge p$.
\end{observation}

The table $T$ could be computed straight from definition in $O(n \log n)$ time. We improve this complexity to $O(n \log \log n)$ by employing Eratosthenes's sieve. The sieve computes, in particular, for each $j=1,\ldots,n$ a list of all distinct prime divisors of $j$. We use these divisors to compute the table via dynamic programming in a right-to-left scan, as shown in Algorithm~\ref{algo:ini}.

\begin{center} 
\begin{minipage}{9cm}
\begin{algorithm}[H]
  $T:=\PREF'$\;
  \vspace{0.1cm}
  \For{$j:=n$ \KwSty{down to} $2$}{
    \ForEach{\mbox{prime divisor} $q$ \mbox{of} $j$}{
       $\TT[j/q] := \min(\TT[j/q],\TT[j])$;
  }}
  \vspace{0.1cm}
  \For{$p:=1$ \KwSty{to} $n$}{
    \lIf{$\TT[p] \ge p$}{$p$ is an initial op-period}
  }
  \vspace{0.1cm}
  \caption{Computing All Initial Op-Periods of $S$}
  \label{algo:ini}
\end{algorithm}
\end{minipage}
\end{center}
  
\begin{theorem}
  All initial op-periods of a string of length $n$ can be computed in $O(n \log \log n)$ time.
\end{theorem}
\begin{proof}
  By Lemma~\ref{lem:PREF}, the $\PREF$ table for the string---hence, the $\PREF'$ table---can be computed in $O(n)$ time. Then we use Algorithm~\ref{algo:ini}. Each prime number $q \leq n$ has at most $\frac{n}{q}$ multiples below $n$. Therefore, the complexity of Eratosthenes's sieve and the number of updates on the table $T$ in the algorithm is $\sum\limits_{q \in \mathit{Primes}, q \le n} \frac{n}{q} = O(n\log\log n)$; see \cite{Analytic}.
\end{proof}

\subsection{Computing Full Op-Periods}
Let us recall the following auxiliary data structure for efficient $\gcd$-computations that was developed in \cite{DBLP:journals/jcss/KociumakaRR17}.
We will only need a special case of this data structure to answer queries for $\gcd(x,n)$.

\begin{fact}[Theorem 4 in \cite{DBLP:journals/jcss/KociumakaRR17}]\label{fct:gcd}
  After $O(n)$-time preprocessing, given any $x,y \in \{1,\ldots,n\}$, the value $\gcd(x,y)$ can be computed in constant time.
\end{fact}

Let $\Div(i)$ denote the set of all positive divisors of $i$. In the case of full op-periods we only need to compute $\TT[p]$ for $p \in \Div(n)$. As in Algorithm~\ref{algo:ini}, we start with $T=\PREF'$. Then we perform a preprocessing phase that shifts the information stored in the array from indices $i \not \in \Div(n)$ to indices $\gcd(i,n) \in \Div(n)$. It is based on the fact that for $d \in \Div(n)$, $d \mid i$ if and only if $d \mid \gcd(i,n)$. Finally, we perform right-to-left processing as in Algorithm~\ref{algo:ini}. However, this time we can afford to iterate over all divisors of elements from $\Div(n)$. Thus we arrive at the pseudocode of Algorithm~\ref{algo:full}.

\begin{center} 
\begin{minipage}{9cm}
\begin{algorithm}[H]
  $T:=\PREF'$\;
  \vspace{0.1cm}
  \For{$i:=1$ \KwSty{to} $n$}{
    $k := \gcd(i,n)$\;
    $\TT[k]:=\min(\TT[k],\TT[i])$\;
  }
  \vspace{0.1cm}
  \ForEach{$i \in \Div(n)$\mbox{ in decreasing order}}{
      \ForEach{$d \in \Div(i)$}{
          $\TT[d]:=\min(\TT[d],\TT[i])$\;
      }
  }
  \vspace{0.1cm}
  \ForEach{$p \in \Div(n)$}{
    \lIf{$\TT[p] \ge p$}{$p$ is a full op-period}
  }
  \caption{Computing All Full Op-Periods of $S$}
  \label{algo:full}
\end{algorithm}
\end{minipage}
\end{center}

\begin{theorem}
  All full op-periods of a string of length $n$ can be computed in $O(n)$ time.
\end{theorem}
\begin{proof}
  We apply Algorithm~\ref{algo:full}. The complexity of the first for-loop is $O(n)$ by Fact~\ref{fct:gcd}. The second for-loop works in $O(n)$ time as the sizes of the sets $\Div(n)$, $\Div(i)$ are $O(\sqrt{n})$ and the elements of these sets can be enumerated in $O(\sqrt{n})$ time as well. 
\end{proof}

\section{Computing Smallest Non-Trivial Initial Op-Period}\label{sec:sm_initial}
If a string is not strictly monotone itself, it has $O(n)$ such op-periods and they can all be computed in $O(n)$ time.
We use this as an auxiliary routine in the computation of the smallest initial op-period that is greater than~1.

\begin{theorem}\label{thm:monotone}
  If a string of length $n$ is not strictly monotone, all of its strictly monotone op-periods can be computed in $O(n)$ time.
\end{theorem}
\begin{proof}
  We show how to compute all the strictly increasing op-periods of a string $S$ that is not strictly monotone itself;
  computation of strictly decreasing and constant op-periods is the same.
  Let $S$ be a string of length $n$ and let us denote $X=\trace(S)$.
  Let $A=\{a_1,\ldots,a_k\}$ be the set of all positions $a_1<\dots<a_k$ in $X$ such that
  $X[i] \ne \texttt{+}$; by the assumption of this theorem, we have that $A \ne \emptyset$.
  This set provides a simple characterization of strictly increasing op-periods of $S$.

  \begin{observation}\label{obs:incr}
    $(p,s)$ is a strictly increasing op-period of a string $S$ that is not strictly monotone itself
    if and only if $a_i = s \pmod{p}$ for all $a_i \in A$.
  \end{observation}

  First, assume that $|A|=1$.
  By Observation~\ref{obs:incr}, each $p=1,\ldots,n$ is an op-period of $S$ with the shift $s=a_1 \bmod p$.
  From now we can assume that $|A|>1$.

  For a set of positive integers $B=\{b_1,\ldots,b_k\}$, by $\gcd(B)$ we denote $\gcd(b_1,\ldots,b_k)$.
  The claim below follows from Fact~\ref{fct:gcd}.
  However, we give a simpler proof.

  \begin{claim}\label{clm:gcd}
    If $B \subseteq \{1,\ldots,n\}$, then $\gcd(B)$ can be computed in $O(n)$ time.
  \end{claim}
  \begin{proof}
    Let $B=\{b_1,\ldots,b_k\}$ and denote $d_i = \gcd(b_1,\ldots,b_i)$.
    We want to compute $d_k$.

    Note that $d_i \mid d_{i-1}$ for all $i=2,\ldots,k$.
    Hence, the sequence $(d_i)$ contains at most $\log n+1$ distinct values.

    Set $d_1=b_1$.  
    To compute $d_i$ for $i \ge 2$, we check if $d_{i-1} \mid b_i$.
    If so, $d_i=d_{i-1}$.
    Otherwise $d_i = \gcd(d_{i-1},b_i)<d_{i-1}$.
    Hence, we can compute $d_i$ using Euclid's algorithm in $O(\log n)$ time.
    The latter situation takes place at most $\log n+1$ times; the conclusion follows.
  \end{proof}

  Consider the set $B=\{a_2-a_1,a_3-a_2,\ldots,a_k-a_{k-1}\}$.
  By Observation~\ref{obs:incr}, $(p,s)$ is a strictly increasing op-period of $S$ if and only if $p \mid \gcd(B)$ and $s = a_1 \bmod p$.
  Thus there is exactly one strictly increasing op-period of each length that divides $\gcd(B)$ and its shift is
  determined uniquely.

  The value $\gcd(B)$ can be computed in $O(n)$ time by Claim~\ref{clm:gcd}.
  Afterwards, we find all its divisors and report the op-periods in $O(\sqrt{n})$ time.
\end{proof}

\noindent
Let us start with the following simple property.

\begin{lemma} \label{l:primitive}
The shape of the smallest non-trivial initial op-period of a string has no shorter non-trivial full op-period.
\end{lemma}
\begin{proof}
A full op-period of the initial op-period of a string $S$ is an initial op-period of $S$.
\end{proof}

\noindent
Now we can state a property of initial op-periods, implied by Theorem~\ref{t:fw1}, that is the basis of the algorithm.

\begin{lemma}\label{lem:key_init}
  If a string of length $n$ has initial op-periods $p>q>1$ such that $p+q<n$ and $\gcd(p,q)=1$, then
  $q$ is strictly monotone.
\end{lemma}
\begin{proof}
  Let us consider three cases. If $n>pq$, then by Theorem~\ref{t:fw1}\ref{aa}, both $p$ and $q$ are strictly monotone. If $2p < n \le pq$, then Theorem~\ref{t:fw1}\ref{ee} implies that $S[1..pq-1]$ is strictly monotone, hence $p$ and $q$ are strictly monotone as well. Finally, if $p+q < n \le 2p$, we have that $q$ is strictly monotone by Theorem~\ref{t:fw1}\ref{ff}.
\end{proof}


\begin{center} 
\begin{minipage}{12.5cm}
\begin{algorithm}[H]
  \caption{Computing the Smallest Non-Trivial Initial Op-Period of $S$}\label{algo5}
  \If{$S$ has a non-trivial strictly monotone op-period}{\Return{smallest such op-period}\Comment*[r]{Theorem~\ref{thm:monotone}}}
  \vspace{0.1cm}
  $p:=$\,the length of the longest monotone prefix of $S$ plus 1\;
  \While{$p\le n$}{
    $k$ := $\LPP_p(S)$\;
    \lIf{$k=n$}{\Return{$p$}}
    $p$ := $\max(p+1,\,k-p-1)$\;
  }
  \Return{$\min(p_{mon},n)$;}
\end{algorithm}
\end{minipage}
\end{center}

\begin{theorem}
  The smallest initial op-period $p>1$ of a string $S$ of length $n$ can be computed in $O(n)$ time.
\end{theorem}
\begin{proof}
  We follow the lines of Algorithm~\ref{algo5}.
  If $S$ is not strictly monotone itself, we can compute the smallest non-trivial strictly monotone initial op-period of $S$ using Theorem~\ref{thm:monotone}.
  Otherwise, the smallest such op-period is 2.
  If $S$ has a non-trivial strictly monotone initial op-period and the smallest such op-period is $q>1$, then none of $2,\ldots,q-1$ is an initial op-period of $S$.
  Hence, we can safely return $q$.

  Let us now focus on the correctness of the while-loop. The invariant is that there is no initial op-period of $S$ that is smaller than $p$. If the value of $k=\LPP_p(S)$ equals $n$, then $p$ is an initial op-period of $S$ and we can safely return it. Otherwise, we can advance $p$ by 1. There is also no smallest initial op-period $p'$ such that $p<p'<k-p-1$. Indeed, Lemma~\ref{lem:key_init} would imply that $p$ is strictly monotone if $\gcd(p,p')=1$ (which is impossible due to the initial selection of $p$) and Theorem~\ref{t:fwd2} would imply an initial op-period of $S[1..p']$ that is smaller than $p'$ and divides $p'$ if $\gcd(p,p')>1$ (which is impossible due to Lemma~\ref{l:primitive}). This justifies the way $p$ is increased.

  Now let us consider the time complexity of the algorithm. The algorithm for strictly monotone op-periods of Theorem~\ref{thm:monotone} works in $O(n)$ time. By Lemma~\ref{lem:LPP}, $k$ can be computed in $O(k/p+1)$ time. If $k \le 3p$, this is $O(1)$. Otherwise, $p$ at least doubles; let $p'$ be the new value of $p$. Then $O(k/p+1)=O((p+p'-1)/p+1)=O(p'+1)$. The case that $p$ doubles can take place at most $O(\log n)$ times and the total sum of $p'$ over such cases is $O(n)$.
\end{proof}

\section{Computing All Op-Periods}\label{sec:general}
An \emph{interval representation} of a set $X$ of integers is
$X=\lb i_1..j_1 \rb \cup \lb i_2..j_2 \rb \cup \dots \cup \lb i_k..j_k \rb$
where $j_1+1 < i_2$, \ldots, $j_{k-1}+1 < i_k$; $k$ is called the \emph{size} of the representation.

Our goal is to compute a \emph{compact representation} of all the op-periods of a string
that contains, for each op-period $p$, an interval representation of the set $\Shifts_p$.

For an integer set $X$, by $X \bmod p$ we denote the set $\{x \bmod p\,:\, x \in X\}$.
The following technical lemma provides efficient operations on interval representations of sets.

\begin{lemma}\label{lem:Compact}
  \begin{enumerate}[label={\rm(\alph*)}]
  \item\label{com:a} Assume that $X$ and $Y$ are two sets with interval representations of sizes $x$ and $y$, respectively.
  Then the interval representation of the set $X \cap Y$ can be computed in $O(x+y)$ time.
  \item\label{com:b} Assume that $X_1,\dots,X_k \subseteq \lb 0..n \rb$ are sets with interval representations
  of sizes $x_1,\dots,x_k$ and $p_1,\ldots,p_k$ be positive integers.
  Then the interval representations of all the sets $X_1 \bmod p_1,\dots,X_k \bmod p_k$ can be computed in $O(x_1+\dots+x_k+k+n)$ time.
  \end{enumerate}
\end{lemma}
\begin{proof}
  To compute $X \cap Y$ in point \ref{com:a}, it suffices to merge the lists of endpoints of intervals in the interval representations of $X$ and $Y$.
  Let $L$ be the merged list.
  With each element of $L$ we store a weight $+1$ if it represents the beginning of an interval and a weight $-1$ if it represents the endpoint of an interval.
  We compute the prefix sums of these weights for $L$.
  Then, by considering all elements with a prefix sum equal to 2 and their following elements in $L$, we can restore the interval representation of $X \cap Y$.

  Let us proceed to point~\ref{com:b}.
  Note that, for an interval $\lb i..j \rb$, the set $\lb i..j \rb \bmod p$ either equals $\lb 0..p-1 \rb$ if $j-i \ge p$, or
  otherwise is a sum of at most two intervals.
  For each interval $\lb i..j \rb$ in the representation of $X_a$, for $a=1,\ldots,k$, we compute the interval representation of $\lb i..j \rb \bmod p_a$.
  Now it suffices to compute the sum of these intervals for each $X_a$.
  This can be done exactly as in point \ref{com:a} provided that the endpoints of the intervals comprising representations of $\lb i..j \rb \bmod p_a$
  are sorted.
  We perform the sorting simultaneously for all $X_a$ using bucket sort \cite{DBLP:books/daglib/0023376}.
  The total number of endpoints is $O(x_1+\dots+x_k)$ and the number of possible values of endpoints is at most $n$.
  This yields the desired time complexity of point~\ref{com:b}.
\end{proof}

\begin{lemma}\label{lem:squares}
  For a string of length $n$, interval representations of the sets $\Squares_p$ for all $1 \le p \le n/2$
  can be computed in $O(n \log n)$ time.
\end{lemma}
\begin{proof}
  Let us define the following two auxiliary sets.
  \begin{align*}
    \LL_p &= \{i \in \lb 1..n-2p+1 \rb\,:\,S[i..i+2p-1]\text{ is a left non-shiftable op-square}\}\\
    \RR_p &= \{i \in \lb 1..n-2p+1 \rb\,:\,S[i..i+2p-1]\text{ is a right non-shiftable op-square}\}.
  \end{align*}
  By Lemma~\ref{lem:op_sq}, all the sets $\LL_p$ and $\RR_p$ can be computed in $O(n \log n)$ time.
  In particular, $\sum_p |\LL_p| = O(n \log n)$.

  Let us note that, for each $p$, $|\LL_p|=|\RR_p|$.
  Thus let $\LL_p=\{\ell_1,\dots,\ell_k\}$ and $\RR_p=\{r_1,\dots,r_k\}$.
  The interval representation of the set $\Squares_p$ is
  $\lb \ell_1..r_1 \rb \cup \dots \cup \lb \ell_k..r_k \rb$.
  Clearly, it can be computed in $O(|\LL_p|)$ time.
\end{proof}


\noindent
We will use the following characterization of op-periods.

\begin{observation}\label{obs:char}
  $p$ is an op-period of $S$ with shift $i$ if and only if all the following conditions hold:
  \begin{enumerate}[label={\rm(\Alph*)}]
    \item $S[i+1+kp..i+(k+2)p]$ is an op-square for every $0 \le k \le (n-2p-i)/p$,
    \item $\LCP(1,p+1) \ge \min(i, n-p)$,
    \item $\LCS(n,n-p) \ge \min((n-i) \bmod p, n-p)$.
  \end{enumerate}
\end{observation}

\begin{theorem}
  A representation of size $O(n \log n)$ of all the op-periods of a string of length $n$
  can be computed in $O(n\log n)$ time.
\end{theorem}
\begin{proof}
  We use Algorithm~\ref{algo:4}.
  The sets $\AA_p$, $\BB_p$, and $\CC_p$ describe the sets of shifts $i$ that satisfy conditions (A), (B), and (C)
  from Observation~\ref{obs:char}, respectively.

  A crucial role is played by the set $\NN_p$ 
  of all positions which are \emph{not} the beginnings of op-squares of
  length $2p$. It is computed as a complement of the set $\Squares_p$.

\begin{center} 
\begin{minipage}{14cm}
\begin{algorithm}[H]
  \caption{Computing a Compact Representation of All Op-Periods}
  \label{algo:4}
  Compute $\Squares_p$ for all $p=1,\ldots,n$\Comment*[r]{Lemma~\ref{lem:squares}}
  \vspace*{0.05cm}
  \For{$p:=1$ \KwSty{to} $n$}{
    $\NN_p:=\lb 1..n-2p+1 \rb \setminus \Squares_p$\;
    $k:=\LCP(1,p+1)$; $\ell:=\LCS(n,n-p)$\;
    \lIf{$k=n-p$}{$\BB_p:=\CC_p:=\lb 1..n \rb$}\lElse{$\BB_p:=\lb 1..k \rb$; $\CC_p:=\lb n-\ell+1..n \rb$}
  }
  \vspace*{0.05cm}
  \For{$p:=1$ \KwSty{to} $n$ simultaneously}{
    $\NN_p:=\{ (x - 1) \bmod p : x\in \NN_p\}$;
    $\BB_p:=\BB_p \bmod p$;\ $\CC_p:=\CC_p \bmod p$\Comment*[r]{Lemma~\ref{lem:Compact}\ref{com:b}}
  }
  \vspace*{0.05cm}
  $\Shifts_1:=\lb 0 \rb$\;
  \For{$p:=2$ \KwSty{to} $n$}{
    $\AA_p:=\lb 0..p-1 \rb \setminus \NN_p$\;
    $\Shifts_p:=\AA_p \cap \BB_p \cap \CC_p$\Comment*[r]{Lemma~\ref{lem:Compact}\ref{com:a}}
  }
  \vspace*{0.1cm}
  \Return{$\Shifts_p$ for $p=1,\dots,n$;}
\end{algorithm}
\end{minipage}
\end{center}

\noindent
  Operations ``$\bmod$'' on sets are performed simultaneously using Lemma~\ref{lem:Compact}\ref{com:b}.
  All sets $\AA_p$, $\BB_p$, $\CC_p$ have $O(n \log n)$-sized representations.
  This guarantees $O(n \log n)$ time.
\end{proof}

\section{Computing Sliding Op-Periods}\label{sec:sliding}
For a string $S$ of length $n$, we define a family of strings $\SHAPE_1,\ldots,\SHAPE_n$ such that
$\SHAPE_k[i]=\shape(S[i..i+k-1])$ for $1 \le i \le n-k+1$.
Note that the characters of the strings are shapes.
Moreover, the total length of strings $\SHAPE_k$ is quadratic in $n$, so we will not compute those strings explicitly.
Instead, we use the following observation to test if two symbols are equal.
 \begin{observation}\label{obs:T}
   $\SHAPE_k[i]=\SHAPE_k[i']$ if and only if $\LCP(i,i') \ge k$.
\end{observation}

Sliding op-periods admit an elegant characterization based on $\SHAPE_k$; see Figure~\ref{fig:slid}.%
\begin{lemma}\label{lem:sliding}
An integer $p$, $1\le p \le n$, is a sliding op-period of $S$ if and only if
$p\le \frac12 n$ and $p$ is a period of $\SHAPE_p$, or $p > \frac12 n$ and $S[1..n-p]\approx S[p+1..n]$.
\end{lemma}
\begin{proof}
If $p$ is a sliding op-period, then $\Shifts_{p}=\lb 0 .. p-1\rb$.
Consequently, Observation~\ref{obs:char} yields that $S[i..i+2p-1]$ is an op-square for every $1 \le i \le n-2p+1$
and that $\LCP(1,p+1)\ge \min(p-1, n-p)$.

If $p\le \frac12 n$, then the former property yields $\SHAPE_{p}[i]=\SHAPE_{p}[i+p]$ for every $1 \le i \le n-2p+1$,
i.e., that $p$ is a period of $\SHAPE_p$.

On the other hand, if $p> \frac12 n$, the latter property implies $\LCP(1,p+1)\ge  n-p$,
i.e., $S[1..n-p]\approx S[p+1..n]$.

For a proof in the other direction, suppose that $p$ satisfies the characterization of Lemma~\ref{lem:sliding}.
If $p>\frac12n$, this yields $\LCP(1,p+1)=n-p=\LCS(n-p,n)$.
Otherwise, $S[i..i+2p-1]$ is an op-square for every $1 \le i \le n-2p+1$
and, in particular, $\LCP(1,p+1)\ge p$ and $\LCS(n-p,n)\ge p$.
In either case the characterization of Observation~\ref{obs:char} yields that $p$ is a sliding op-period.
\end{proof}

  \begin{figure}[t]
    \begin{center}
      \begin{tikzpicture}[scale=0.4]

\newcommand{\xstep}{1.1}
\newcommand{\ystep}{0.9}

\begin{scope}
  \foreach \x in {1,...,18}{\draw[densely dotted] (\x * \xstep,0 * \ystep) -- (\x * \xstep, 12 * \ystep);}
  \foreach \y in {0,...,12}{\draw[densely dotted] (1 * \xstep,\y * \ystep) -- (18 * \xstep, \y * \ystep);}
  \foreach \x/\y in {1/0,2/12,3/6,4/1,5/11,6/6,7/2,8/10,9/6,10/3,11/9,12/6,13/4,14/8,15/6,16/5,17/7,18/6}{
    \draw[very thick] (\x * \xstep,\y * \ystep) circle (0.07cm);
  }

  \foreach \x/\y in {1/0,4/1,7/2,10/3,13/4,16/5,18/6}{
    \draw (\x * \xstep, \y * \ystep) node[below right=-0.1cm] {\small \y};
  }
  \foreach \x/\y in {3/6,6/6,9/6,12/6,15/6}{
    \draw (\x * \xstep, \y * \ystep) node[below left=-0.1cm] {\small \y};
  }
  \foreach \x/\y in {11/9,14/8,17/7}{
    \draw (\x * \xstep, \y * \ystep) node[above right=-0.1cm] {\small \y};
  }
  \draw (2 * \xstep, 12 * \ystep) node[above right=-0.1cm] {\small $1\!2$};
  \draw (5 * \xstep, 11 * \ystep) node[above right=-0.1cm] {\small $1\!1$};
  \draw (8 * \xstep, 10 * \ystep) node[above right=-0.1cm] {\small $1\!0$};

  \draw (9.5 * \xstep,-0.5 * \ystep) node {\small $S$};
  \draw (1 * \xstep,0 * \ystep) -- (2 * \xstep,12 * \ystep) -- (3 * \xstep,6 * \ystep) -- (4 * \xstep,1 * \ystep) -- (5 * \xstep,11 * \ystep) -- (6 * \xstep,6 * \ystep) -- (7 * \xstep,2 * \ystep) -- (8 * \xstep,10 * \ystep) -- (9 * \xstep,6 * \ystep) -- (10 * \xstep,3 * \ystep) -- (11 * \xstep,9 * \ystep) -- (12 * \xstep,6 * \ystep) -- (13 * \xstep,4 * \ystep) -- (14 * \xstep,8 * \ystep) -- (15 * \xstep,6 * \ystep) -- (16 * \xstep,5 * \ystep) -- (17 * \xstep,7 * \ystep) -- (18 * \xstep,6 * \ystep);
\end{scope}

\begin{scope}[xshift=21cm,yshift=6cm]
  \foreach \x in {1,...,6}{\draw[densely dotted] (\x * \xstep,1 * \ystep) -- (\x * \xstep, 5 * \ystep);}
  \foreach \y in {1,...,5}{\draw[densely dotted] (1 * \xstep,\y * \ystep) -- (6 * \xstep, \y * \ystep);}
  \foreach \x/\y in {1/1,2/5,3/3,4/2,5/4,6/3}{
    \draw[very thick] (\x * \xstep,\y * \ystep) circle (0.07cm);
  }

  \foreach \x/\y in {1/1,4/2,6/3}{
    \draw (\x * \xstep, \y * \ystep) node[below right=-0.1cm] {\small \y};
  }
  \draw (3 * \xstep, 3 * \ystep) node[below left=-0.1cm] {\small 3};
  \foreach \x/\y in {2/5,5/4}{
    \draw (\x * \xstep, \y * \ystep) node[above right=-0.1cm] {\small \y};
  }

  \draw (3.5 * \xstep,0.5 * \ystep) node {\small $A$};
  \draw (1 * \xstep,1 * \ystep) -- (2 * \xstep,5 * \ystep) -- (3 * \xstep,3 * \ystep) -- (4 * \xstep,2 * \ystep) -- (5 * \xstep,4 * \ystep) -- (6 * \xstep,3 * \ystep);
\end{scope}

\begin{scope}[xshift=21cm,yshift=0cm]
  \foreach \x in {1,...,6}{\draw[densely dotted] (\x * \xstep,1 * \ystep) -- (\x * \xstep, 5 * \ystep);}
  \foreach \y in {1,...,5}{\draw[densely dotted] (1 * \xstep,\y * \ystep) -- (6 * \xstep, \y * \ystep);}
  \foreach \x/\y in {1/5,2/3,3/1,4/4,5/3,6/2}{
    \draw[very thick] (\x * \xstep,\y * \ystep) circle (0.07cm);
  }

  \foreach \x/\y in {3/1,6/2}{
    \draw (\x * \xstep, \y * \ystep) node[below right=-0.1cm] {\small \y};
  }
  \draw (5 * \xstep, 3 * \ystep) node[below left=-0.1cm] {\small 3};
  \foreach \x/\y in {1/5,2/3,4/4}{
    \draw (\x * \xstep, \y * \ystep) node[above right=-0.1cm] {\small \y};
  }

  \draw (3.5 * \xstep,-0.5 * \ystep) node {\small $B$};
  \draw (1 * \xstep,5 * \ystep) -- (2 * \xstep,3 * \ystep) -- (3 * \xstep,1 * \ystep) -- (4 * \xstep,4 * \ystep) -- (5 * \xstep,3 * \ystep) -- (6 * \xstep,2 * \ystep);
\end{scope}

\begin{scope}[xshift=28.5cm,yshift=3cm]
  \foreach \x in {1,...,6}{\draw[densely dotted] (\x * \xstep,1 * \ystep) -- (\x * \xstep, 5 * \ystep);}
  \foreach \y in {1,...,5}{\draw[densely dotted] (1 * \xstep,\y * \ystep) -- (6 * \xstep, \y * \ystep);}
  \foreach \x/\y in {1/3,2/1,3/5,4/3,5/2,6/4}{
    \draw[very thick] (\x * \xstep,\y * \ystep) circle (0.07cm);
  }

  \foreach \x/\y in {2/1,5/2}{
    \draw (\x * \xstep, \y * \ystep) node[below right=-0.1cm] {\small \y};
  }
  \draw (4 * \xstep, 3 * \ystep) node[below left=-0.1cm] {\small 3};
  \foreach \x/\y in {1/3,3/5,6/4}{
    \draw (\x * \xstep, \y * \ystep) node[above right=-0.1cm] {\small \y};
  }

  \draw (3.5 * \xstep,0.5 * \ystep) node {\small $C$};
  \draw (1 * \xstep,3 * \ystep) -- (2 * \xstep,1 * \ystep) -- (3 * \xstep,5 * \ystep) -- (4 * \xstep,3 * \ystep) -- (5 * \xstep,2 * \ystep) -- (6 * \xstep,4 * \ystep);
\end{scope}

\end{tikzpicture}
    \end{center}
    \caption{
      A string $S=0\,1\!2\,6\,1\,1\!1\,6\,2\,1\!0\,6\,3\,9\,6\,4\,8\,6\,5\,7\,6$ is graphically illustrated above (the $i$th point has coordinates $(i,S[i])$).
  We have $\SHAPE_6=ABCABCABCA$, where $A=1\,5\,3\,2\,4\,3$, $B=5\,3\,1\,4\,3\,2$, and $C=3\,1\,5\,3\,2\,4$.
  The shortest period of $\SHAPE_6$ is 3.
  Hence, 6 is a sliding op-period of $S$.
  Moreover, Lemma~\ref{lem:obs}\ref{it:obs:right} implies that 3 is a period of $\SHAPE_3$, hence a sliding op-period of $S$.
    }\label{fig:slid}
  \end{figure}
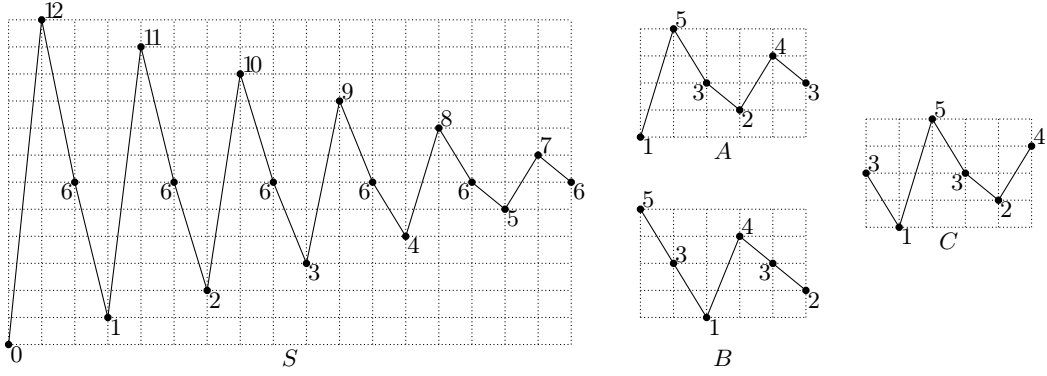

For a string $X$, we denote the shortest period of $X$ by $\per(X)$.
\begin{lemma}\label{lem:obs}
Suppose that $p=\per(\SHAPE_k[1..\ell])<\ell$. Then
\begin{enumerate}[label={\rm(\alph*)}]
  \item\label{it:obs:down} $p$ is also a period of $\SHAPE_{k'}[1..\ell+k-k']$ for $1\le k'\le k$,
  \item\label{it:obs:right} $q = \per(\SHAPE_k[1..\ell+1])$ satisfies $p=q$ or $p+q> \ell$.
\end{enumerate}
\end{lemma}
\begin{proof}
Observe that $\SHAPE_k[i]=\SHAPE_k[i']$ is equivalent to $S[i..i+k-1]\approx S[i'..i'+k-1]$.
The relation $\approx$ is hereditary, so 
$S[j..j+k'-1]\approx S[j'..j'+k'-1]$ if $i \le j$, $j+k' \le i+k$, and $j-i=j'-i'$.
Thus, $$\SHAPE_{k'}[i..i+k-k']=\SHAPE_{k'}[i'..i'+k-k']$$ for each $k'\le k$.

Hence, $\SHAPE_k[1..\ell-p]=\SHAPE_k[p+1..\ell]$ implies
$\SHAPE_{k'}[1..\ell+k-k'-p]=\SHAPE_{k'}[p+1..\ell+k-k']$, which gives \ref{it:obs:down}.

For a proof of \ref{it:obs:right}, we observe that $p$ and $q$ are both periods of $\SHAPE_{k}[1..\ell]$.
If $p+q \le \ell$, then Periodicity Lemma implies $p \mid q$.
Thus, $\SHAPE_{k}[\ell+1]=\SHAPE_{k}[\ell+1-q]=\SHAPE_{k}[\ell+1-p]$,
i.e., $p=q$.
\end{proof}

\noindent We introduce a two-dimensional table $\PP$, where:

\centerline{$\PP[k, \ell] = \per(\SHAPE_k[1..\ell])$ if  $\per(\SHAPE_{k}[1..\ell])\le \frac13 \ell$,
and $\PP[k,\ell]=\bot$ (undefined) otherwise.}

\noindent
The size of $\PP$ is quadratic in $n$. However, Algorithm~\ref{algo2} computes $\PP$
column after column, keeping only the current column 
$\TT=\PP[\cdot,\ell]$.
The total number of differences between consecutive columns is linear.
Hence, any requested $O(n)$ values $\PP[k,\ell]$ can be computed in $O(n)$ time.
We also use an analogous table $\PP^R$ for the reverse string $S^R$.

\begin{center} 
\begin{algorithm}[H]
\LinesNumbered
\caption{Computation of $\PP[\cdot,\ell]$ from $\PP[\cdot,\ell-1]$}\label{algo2}
$\TT[1..n] := [\bot,\ldots,\bot]$; $t := 1$; $\ell' := 3$\;
\For{$\ell := 1$ \KwSty{to} $n$}{
\If{$t > 1$ \KwSty{and} $\SHAPE_{t-1}[\ell] \ne \SHAPE_{t-1}[\ell- \TT[t-1]]$}{\label{line:if}
$t := t-1$; $\TT[t]:=\bot$; $\ell' := 2\ell$\;\label{line:decr}
}
\If{$\ell \ge \ell'$}{\label{line:guard}
\While{$\per(\SHAPE_{t}[1..\ell])=\frac13\ell$}{\label{line:while}
 $\TT[t] := \frac13\ell$; $t := t+1$;    $\ell' := 2\ell$\;\label{line:incr}
}
}
\Comment{\rm Invariant:$\;\TT[k] = \PP[k,\ell]$
  , $t = \min\{k : \TT[k]=\bot\}$, and $\per(\SHAPE_{t}[1..\ell]) \ge \frac13 \ell'$.}
}
\end{algorithm}
\end{center}

\begin{lemma}\label{lem:cor}
Algorithm~\ref{algo2} is correct, that is, it satisfies the invariant.
\end{lemma}
\begin{proof}
First, observe that the invariant is satisfied after the first iteration.
This is because $\per(\SHAPE_k[1..1])=1$ for each $k$ and  the initial values are not changed during this iteration.

Thus, our task is to prove that the invariant is preserved after each subsequent $\ell$th iteration.
Let $t=\min\{k : \PP[k,\ell-1]=\bot\}$ and $t'=\min\{k : \PP[k,\ell]=\bot\}$.

First, we consider the values $\PP[k,\ell]$ for $k < t$.
For this,  we assume $t>1$ and denote $p = \PP[t-1,\ell-1]$.
Since $p$ is a period of $\SHAPE_{t-1}[1..\ell-1]$, Lemma~\ref{lem:obs}\ref{it:obs:down} yields that $p$ is also a period of $\SHAPE_{k}[1..\ell]$ for $k<t-1$.
We apply Lemma~\ref{lem:obs}\ref{it:obs:right} for $p'=\per(\SHAPE_{k}[1..\ell-1])$. 
Since $p'+p \le \ell-1$, we conclude that $p'=\per(\SHAPE_{k}[1..\ell])$, i.e., $\PP[k,\ell-1]=p'=\PP[k,\ell]$.
Now, we consider the value $\PP[t-1,\ell]$. Lemma~\ref{lem:obs}\ref{it:obs:right}, applied for $p=\per(\SHAPE_{t-1}[1..\ell-1])$
and $q=\per(\SHAPE_{t-1}[1..\ell])$, yields $p=q$ or $p+q\ge \ell$.
To verify the first case, we check whether $\SHAPE_{t-1}[\ell]=\SHAPE_{t-1}[\ell-p]$.
In the second case, we conclude that $q \ge \frac23\ell$, so $\PP[t-1,\ell]=\bot$ (and $\ell':=2\ell$ is also set correctly).

Next, we consider the values $\PP[k,\ell]$ for $k\ge t$. Since $\PP[k,\ell-1]=\bot$, we have $\PP[k,\ell]=\bot$ or $\PP[k,\ell]=\frac13\ell$.
More precisely, $\PP[k,\ell]=\bot$ for $k\ge t'$ and $\PP[k,\ell]=\frac13\ell$ for $t \le k < t'$.
Thus, we check if $\per(\SHAPE_{k}[1..\ell])=\frac13\ell$ for subsequent values $k\ge t$.
Since $\per(\SHAPE_{t}[1..\ell])\ge \frac13 \ell'$, no verification is needed if $\ell < \ell'$.
To complete the proof, we need to show that the update $\ell':=2\ell$ is valid if $t'>t$.
For a proof by contradiction suppose that $r := \per(\SHAPE_{t'}[1..\ell]) < \frac23 \ell$.
By Lemma~\ref{lem:obs}\ref{it:obs:down}, $r$ is a period of $\SHAPE_{t}[1..\ell]$. Since $r+\frac13\ell \le \ell$,
Periodicity Lemma yields $\frac13 \ell \mid r$, and thus $r=\frac13\ell$, which contradicts the definition of $t'$.
\end{proof}

\begin{lemma}\label{lem:complexity}
Algorithm~\ref{algo2} can be implemented in time $O(n)$ plus the time to answer $O(n)$ $\LCP$ queries in $S$.
\end{lemma}
\begin{proof}
First, observe that each line is executed $O(n)$ times. 
Indeed, we always have $t\le n$ and $t$ is decremented at most $n$ times in Line~\ref{line:decr}, so the number of increments in Line~\ref{line:incr}
is $O(n)$.

Each instruction takes constant time except for the conditions in Lines~\ref{line:if} and~\ref{line:while}.
The test in Line~\ref{line:if} can be implemented using a single $\LCP$ query (due to Observation~\ref{obs:T}).
Checking the condition in Line~\ref{line:while} requires a more careful implementation exploiting the structure of the queries. 

Suppose that the variable $t$ has been changed in iterations $\ell_1<\dots<\ell_m$. For consistence, we also define $\ell_0 = 1$ and $\ell_{m+1}=n+1$.
Consider a \emph{phase}, consisting of iterations $\ell\in \lb \ell_i+1 .. \ell_{i+1} \rb$.
Observe that $\ell'\ge 2\ell_i$ during the $i$th phase, so Line~\ref{line:while} is executed only during phases such that $\ell_{i+1}\ge 2\ell_i$.

Consider such a phase with $t=\min\{k : \PP[k,\ell_i]=\bot\}$.
We use the Knuth--Morris--Pratt algorithm~\cite{DBLP:books/daglib/0020103} to determine $\per(\SHAPE_{t}[1..\ell])$ for subsequent values $\ell\ge 2\ell_i$.
This takes $O(\ell_{i+1})$ time and, additionally, requires $O(\ell_{i+1})$ symbol equality checks within $\SHAPE_t$, which are implemented based on Observation~\ref{obs:T} using $\LCP$ queries.

If we learn that $\per(\SHAPE_{t}[1..\ell])=\frac13\ell$, we conclude that $\ell=\ell_{i+1}$.
We compute the largest $t'$ such that $\per(\SHAPE_{t'}[1..\ell])= \frac13\ell$ and for the subsequent values $k\ge t$ we simply verify if $k\le t'$ to check if $\per(\SHAPE_{k}[1..\ell])= \frac13\ell$.
Due to Observation~\ref{obs:T}, we have $$t'= \min\{\LCP(i,i+\tfrac13\ell)\,:\, 
1 \le i \le \tfrac23\ell\},$$ so $t'$ can be determined in $O(\ell_{i+1})$ time plus the time to answer $O(\ell_{i+1})$ $\LCP$-queries.

The $i$th phase makes $O(\ell_{i+1})$ steps if $\ell_{i+1}\ge 2\ell_i$, and $O(\ell_{i+1}-\ell_i)$ in general.
The overall running time is therefore $O(n)$ plus the time to answer $O(n)$ $\LCP$-queries.
\end{proof}

\begin{center} 
\begin{algorithm}[H]
\LinesNumbered
\caption{Computing the sliding op-periods $p \le \frac12 n$}\label{algo3}
$p := 1$\;
\While{$p\le \frac12 n$}{
\If{$(q := \PP[p, n-2p+1])=\PP^R[p, n-2p+1] \ne \bot$}{
\lIf{$p$ is a period of $\SHAPE_{p}[1..p+q]$}{report $p$}\label{line:per1}
$p := \min\{p' > p  : p' \text{ is a period of }\SHAPE_p[1..p+2q]\}$\label{line:nxt1}
}
\lElseIf{$\PP[p, \lceil\frac34(n-2p+1)\rceil]=\PP^R[p, \lceil\frac34(n-2p+1)\rceil] \ne \bot$}{$p := p+1$\label{line:2}}
\Else{
\lIf{$p$ is a period of $\SHAPE_p$}{report $p$}\label{line:per2}
$p := \min\{p' > p  : p' \text{ is a period of }\SHAPE_p\}$;\label{line:nxt2}
}
}
\end{algorithm}
\end{center}

\begin{lemma}\label{lem:cor2}
Algorithm~\ref{algo3} is correct, that is, it reports all sliding op-periods $p\le \frac12n$ of $S$.
\end{lemma}
\begin{proof}
Let $p_i$ be the value of $p$ at the beginning of the $i$th iteration of the while-loop and let $\ell_i = n-2p_i+1$.
We shall prove that $p_i$ is reported if and only if it is a sliding op-period
and that there is no sliding op-period strictly between $p_{i}$ and $p_{i+1}$.

First, suppose that $q = \per(\SHAPE_{p_i}[1..\ell_i])=\per(\SHAPE_{p_i}[p_i+1..p_i+\ell_i])\le \frac13\ell_i$, i.e., we are in the first branch.
If $\SHAPE_{p_i}[1..q]=\SHAPE_{p_i}[p_i+1..p_i+q]$, then we must have
$\SHAPE_{p_i}[1..\ell_i] = \SHAPE_{p_i}[p_i+1..p_i+\ell_i]$, i.e., $p_i$ is a period of $\SHAPE_{p_i}=\SHAPE_{p_i}[1..p_i+\ell_i]$ and $p_i$ is a sliding op-period due to Lemma~\ref{lem:sliding}.
Moreover, any sliding op-period $p'>p_i$ must be a period of $\SHAPE_{p_i}$ (and, in particular, of $\SHAPE_{p_i}[1..p_i+2q]$) due to Lemma~\ref{lem:obs}\ref{it:obs:down}.
Consequently, $p'\ge p_{i+1}$, as claimed.

In the second branch we only need to prove that $\SHAPE_{p_i}[1..\ell_i]\ne \SHAPE_{p_i}[p_i+1..p_i+\ell_i]$.
For a proof by contradiction, suppose that we have an equality.
The condition from Line~\ref{line:2} means that the length-$\lceil\frac34 \ell_i\rceil$ prefix and suffix of $\SHAPE_{p_i}[1..\ell_i]=\SHAPE_{p_i}[p_i+1..p_i+\ell_i]$ 
has the common shortest period $q \le \frac13 \lceil\frac34 \ell_i\rceil \le \lceil\frac14 \ell_i\rceil$.
The prefix and the suffix overlap by at least $\lceil\frac12 \ell_i \rceil$ characters, so we actually have $q=\per(\SHAPE_{p_i}[1..\ell_i])=\per(\SHAPE_{p_i}[p_i+1..p_i+\ell_i])$.
Hence, in that case we would be in the first branch.

Finally, in the third branch we directly use  Lemma~\ref{lem:sliding} to check if $p_i$ is a sliding op-period.
Moreover, if $p'>p_i$ is also a sliding op-period, then $p'$ is a period of $\SHAPE_{p_i}$, i.e., $p'\ge p_{i+1}$.
\end{proof}

\begin{lemma}\label{lem:time}
Algorithm~\ref{algo3} can be implemented in time $O(n)$ plus the time to answer $O(n)$ $\LCP$ and $\LCS$ queries in $S$.
\end{lemma}
\begin{proof}
It suffices to bound the time complexity assuming that each $\LCP$ and $\LCS$ query takes unit time.

First, we observe that $\PP[k,\ell]$ and $\PP^R[k,\ell]$ is used only for $\ell=n-2k+1$ or $\ell=\ceil{\frac34 (n-2k+1)}$.
These $O(n)$ values can be computed in $O(n)$ time using Algorithm~\ref{algo2}.

The condition in Line~\ref{line:per1} can be verified using $O(q)$ equality checks in $\SHAPE_p$,
whereas in Line~\ref{line:nxt1}, it suffices to compute the border table of $\SHAPE_p[1..2q]\SHAPE_p[p+1..p+2q]$,
which also takes $O(q)$ time and equality checks in $\SHAPE_p$.
By a similar argument, the third branch can be implemented in $O(|\SHAPE_p|)=O(n-2p+1)$ time, whereas the second branch clearly takes $O(1)$ time.

In order to prove that the total running time is $O(n)$, we introduce a potential function.
Let $p_i$ be the value of the variable $p$ at the beginning of the $i$th iteration,
let $p'_i=\min \{p' > p_i : p'\text{ is a period of }\SHAPE_{p_i}\}$.
Note that $p_i < p'_i \le |\SHAPE_{p_i}|$ due to $p_i\le \frac12 n$.
Moreover, $p_{i+1}>p_i$ and $p'_{i+1}\ge p'_{i}$ by Lemma~\ref{lem:obs}\ref{it:obs:down}.

Our potential function is 
$$\phi_i=p_i + p'_i,$$ 
i.e., we shall prove that the running time of the $i$th iteration is $O(\phi_{i+1}-\phi_i)$.

The running time of the first branch is $O(q)$, so we shall prove that $\phi_{i+1}-\phi_i \ge q$.
Assume to the contrary that $p'_{i+1}-p'_i+p_{i+1}-p_i<q$. This yields that $p_{i+1}<p_{i}+q$ and $p'_{i+1}<p'_i+q$.
The first condition implies that 
$$\SHAPE_{p_i}[1..q]=\SHAPE_{p_i}[p_{i+1}+1..p_{i+1}+q].$$
Since $q$ is a period of $\SHAPE_{p_i}[p_{i+1}+1..p_i+\ell_i]$ and of $\SHAPE_{p_i}[1..\ell_i]$, we conclude that $p_{i+1}$ is a period of $\SHAPE_{p_i}$
so $p_{i+1}=p'_i$.
Due to Lemma~\ref{lem:obs}\ref{it:obs:down}, the condition $p'_{i+1}<p'_i+q=p_{i+1}+q$ implies 
$$\SHAPE_{p_i}[1..q]=\SHAPE_{p_i}[p'_{i+1}+1..p'_{i+1}+q].$$
This gives a non-trivial occurrence of $\SHAPE_{p_i}[1..q]$ in $\SHAPE_{p_i}[p_{i+1}+1..p_{i+1}+2q]=\SHAPE_{p_i}[1..q]^2$,
which contradicts the primitivity of $\SHAPE_{p_i}[1..q]$.

The running time of the second branch is $O(1)$ and we indeed have $\phi_{i+1}-\phi_i\ge 1$.

In the third branch, the running time is $O(n-2p+1)$ and we shall prove that 
$$\phi_{i+1}-\phi_i \ge \tfrac14 \ell_i.$$
For a proof by contradiction, suppose that $\phi_{i+1}-\phi_i < \frac14 \ell_i$.
In this branch we have $p_{i+1}=p'_i$, so $\phi_{i+1}-\phi_i=p'_{i+1}-p_i$.
By Lemma~\ref{lem:obs}\ref{it:obs:down}, both $p'_{i}$ and $p'_{i+1}$ are periods of $\SHAPE_{p_i}$.
Hence, $p'_{i+1}-p'_i$ is a period of
$$\SHAPE_{p_i}[p'_i+1..n-p_i+1] = \SHAPE_{p_i}[1..n-p'_i-p_i+1].$$
In particular,
$$\per(\SHAPE_{p_i}[1..n-p'_i-p_i+1])=\per(\SHAPE^R_{p_i}[1..n-p'_i-p_i+1])<\tfrac14 \ell_i.$$
Since $p'_i-p_i \le p'_{i+1}-p_i\le \frac14\ell_i$, we have
$n-p'_i-p_i+1 = \ell_i-p'_i+p_i \ge \frac34 \ell_i$.
Consequently,
$$\per(\SHAPE_{p_i}[1..\ceil{\tfrac34 \ell_i}])=\per(\SHAPE^R_{p_i}[1..\ceil{\tfrac34 \ell_i}])<\tfrac14 \ell_i.$$
Hence, $\PP[p_i,\lceil\frac34 \ell_i\rceil]$ and $\PP^R[p_i,\lceil\frac34 \ell_i\rceil]$ are both equal to
this common value.
This is a contradiction, because in that case we would be in the second branch.
This completes the proof.
\end{proof}

\begin{theorem}
  All sliding op-periods of a string of length $n$ can be computed in $O(n)$ space and
    $O(n \log \log n)$ expected time or $O(n \log^2 \log n/\log \log \log n)$ worst-case time.
\end{theorem}
\begin{proof}
First, we apply Lemma~\ref{lem:LCP} so that $\LCP$ and $\LCS$ queries can be answered in $O(1)$ time.
Next, we run Algorithm~\ref{algo3} to report sliding op-periods $p\le \frac12 n$.
Then, we iterate over $p>\frac12n$ and report $p$ if $\LCP(1,p+1)=n-p$.
Correctness follows from Lemmas~\ref{lem:cor2} and~\ref{lem:sliding}.
The overall time is $O(n)$ (Lemma~\ref{lem:time}) plus the preprocessing time of Lemma~\ref{lem:LCP}.
\end{proof}

\paragraph{Acknowledgements.}
A part of this work was done during the workshop ``StringMasters in Warsaw 2017'' that was sponsored by
the Warsaw Center of Mathematics and Computer Science.
The authors thank the participants of the workshop, especially Hideo Bannai and Shunsuke Inenaga,
for helpful discussions.

\bibliographystyle{plain}
\bibliography{op_periods}

\end{document}